\documentclass[aps,pra,reprint,amsmath,amssymb,nofootinbib,floatfix,superscriptaddress]{revtex4-2}

\usepackage{graphicx}
\usepackage{booktabs}
\usepackage{amsthm}
\usepackage{dsfont}
\usepackage{algpseudocode}
\makeatletter
\newcounter{algorithm}
\renewcommand{\thealgorithm}{\arabic{algorithm}}
\def\fps@algorithm{tbp}
\def\ftype@algorithm{4}
\def\ext@algorithm{loa}
\providecommand{\fname@algorithm}{Algorithm}
\def\fnum@algorithm{\fname@algorithm\nobreakspace\thealgorithm}
\newenvironment{algorithm}{%
  \let\alg@mkcap\@makecaption
  \def\@makecaption##1##2{%
    \hrule height.8pt depth0pt\kern2pt
    {\setlength\abovecaptionskip{2pt}%
     \alg@mkcap{##1}{##2}}%
    \kern2pt\hrule height.4pt\kern2pt}%
  \@float{algorithm}}{%
  \kern2pt\hrule height.4pt\relax
  \end@float}
\makeatother
\usepackage{tikz}
\usepackage{orcidlink}
\providecommand{\orcid}[1]{\orcidlink{#1}}
\usepackage{hyperref}
\hypersetup{hidelinks}
\usetikzlibrary{arrows.meta, positioning, fit, calc, shapes.symbols, patterns, backgrounds}

\newtheorem{theorem}{Theorem}
\newtheorem{lemma}{Lemma}
\newtheorem{proposition}{Proposition}

\newcommand{\ket}[1]{\lvert #1 \rangle}
\newcommand{\bra}[1]{\langle #1 \rvert}

\newcommand{\diag}{\mathrm{diag}}
\newcommand{\good}{\mathrm{good}}
\newcommand{\bad}{\mathrm{bad}}
\newcommand{\poly}{\mathrm{poly}}
\newcommand{\Ham}{\operatorname{Ham}}

\newcommand{\casql}{Laboratory of Quantum Information, University of Science and Technology of China, Hefei, Anhui, 230026, China}
\newcommand{\casex}{Anhui Province Key Laboratory of Quantum Network, University of Science and Technology of China, Hefei 230026, China}
\newcommand{\aihf}{Institute of Artificial Intelligence, Hefei Comprehensive National Science Center, Hefei, Anhui, 230088, China}
\newcommand{\origin}{Origin Quantum Computing Technology (Hefei) Co., Ltd., Hefei, Anhui, 230026, China}

\begin{document}

\title{Classical simulation of amplitude-damped bucket-brigade quantum random access memories via predictable branch evolution}

\author{Zhao-Yun Chen\,\orcid{0000-0002-5181-160X}}
\email{chenzhaoyun@iai.ustc.edu.cn}
\affiliation{\aihf}

\author{Ming-Yang Tan}
\affiliation{\casql}

\author{Sheng Zhang}
\affiliation{\casql}

\author{Peng Wang}
\affiliation{\casql}

\author{Yun-Jie Wang\,\orcid{0009-0000-5347-5286}}
\affiliation{\casql}

\author{Guo-Ping Guo\,\orcid{0000-0002-2179-9507}}
\affiliation{\casql}
\affiliation{\casex}
\affiliation{\aihf}
\affiliation{\origin}

\begin{abstract}
Classical simulation of quantum random access memory (QRAM) under noise can be accelerated dramatically by branch pruning: noise histories are sampled as trajectory ensembles, and good branches, whose routing paths avoid every sampled fault, need not be evolved because their final states are analytically predictable. Under amplitude damping this predictability is not automatic, because the no-jump operator $K_0$ acts on every branch at every time slice. Here we give a complete account of damping-channel predictability in bucket-brigade QRAM simulation, for both qutrit and qubit encodings. For the qutrit encoding, $K_0$ is diagonal in the node basis, and the wait state $\ket{W}$ is its fixed point, so a good branch follows the noiseless orbit up to a scalar attenuation. In the qubit encoding, phase-kickback memory fetching between two Hadamard walls makes $K_0$ no longer diagonal, yet the resulting structure is exactly solvable in closed form. Conditional on no jump, the in-branch infidelity of a good branch is second order in the damping rate $\gamma$, with coherent leakage $\sim n^2\gamma^2/4$ per data qubit, and the closed form captures it without approximation. We also revise the pruning criterion for the qubit encoding. The resulting algorithm evolves only the bad branches plus one reference branch and shares every other cost with the full evolution it replaces. It is validated by end-to-end benchmarks with pruned-over-full speedups up to $285\times$ and by trajectory-level tests confirming every closed form to machine precision.
\end{abstract}


\maketitle

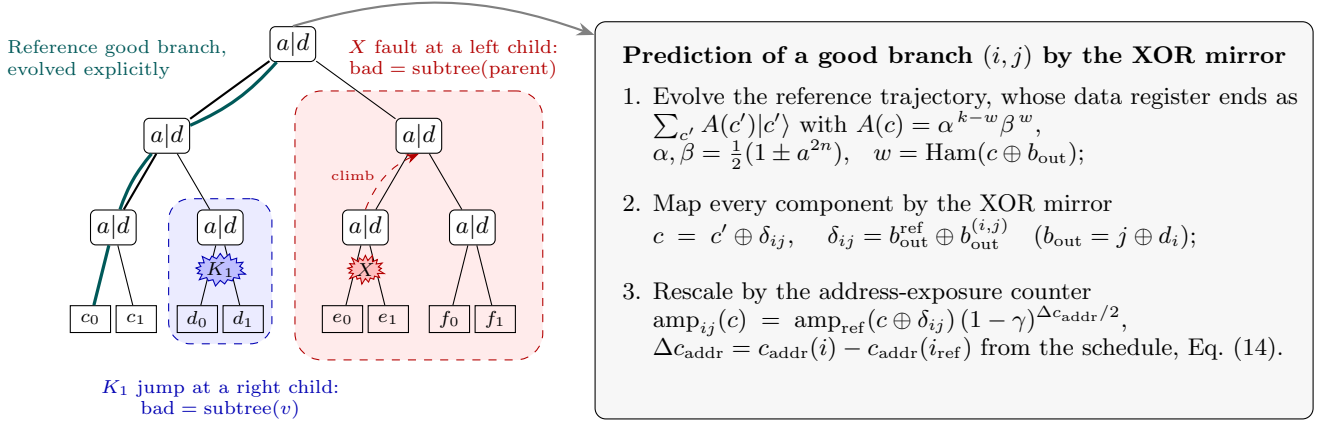
\begin{figure*}[t]
\centering
\resizebox{0.97\textwidth}{!}{%
\begin{tikzpicture}[
  font=\small,
  qnode/.style={draw, rounded corners=2pt, minimum width=17pt, minimum height=13pt, inner sep=1pt, fill=white},
  cell/.style={draw, rectangle, minimum width=15pt, minimum height=10pt, inner sep=0pt, font=\scriptsize},
  faultx/.style={starburst, starburst point height=2.6pt, minimum size=11pt, inner sep=0pt, draw=red!70!black, fill=red!25, font=\scriptsize},
  faultk/.style={starburst, starburst point height=2.6pt, minimum size=11pt, inner sep=0pt, draw=blue!70!black, fill=blue!25, font=\scriptsize},
  >=Stealth]

\node[qnode] (R) at (2.65, 7.6) {$a|d$};
\node[qnode] (A) at (1.0, 6.4) {$a|d$};
\node[qnode] (B) at (4.3, 6.4) {$a|d$};
\node[qnode] (C) at (0.3, 5.2) {$a|d$};
\node[qnode] (D) at (1.7, 5.2) {$a|d$};
\node[qnode] (E) at (3.6, 5.2) {$a|d$};
\node[qnode] (F) at (5.0, 5.2) {$a|d$};
\draw[thick] (R) -- (A); \draw (R) -- (B);
\draw[thick] (A) -- (C); \draw (A) -- (D);
\draw[very thick, teal!70!black] (R) to[bend left=12] (A);
\draw (B) -- (E); \draw (B) -- (F);

\node[cell] (c0) at (0.0, 4.0) {$c_0$};
\node[cell] (c1) at (0.60, 4.0) {$c_1$};
\node[cell] (d0) at (1.40, 4.0) {$d_0$};
\node[cell] (d1) at (2.00, 4.0) {$d_1$};
\node[cell] (e0) at (3.30, 4.0) {$e_0$};
\node[cell] (e1) at (3.90, 4.0) {$e_1$};
\node[cell] (f0) at (4.70, 4.0) {$f_0$};
\node[cell] (f1) at (5.30, 4.0) {$f_1$};
\draw (C) -- (c0); \draw (C) -- (c1);
\draw (D) -- (d0); \draw (D) -- (d1);
\draw (E) -- (e0); \draw (E) -- (e1);
\draw (F) -- (f0); \draw (F) -- (f1);

\draw[very thick, teal!70!black] (A) to[bend right=10] (C);
\draw[very thick, teal!70!black] (C) -- (c0);
\node[teal!70!black, font=\scriptsize, align=left, anchor=west] at (-1.2, 7.4)
  {Reference good branch,\\ evolved explicitly};

\node[faultx] (X) at (3.6, 4.65) {$X$};
\node[faultk] (K) at (1.7, 4.65) {$K_1$};

\begin{scope}[on background layer]
  \node[draw=red!70!black, dashed, rounded corners=10pt, fill=red!8,
        inner xsep=10pt, inner ysep=10pt, fit=(B)(E)(F)(e0)(e1)(f0)(f1)] (badx) {};
  \node[draw=blue!70!black, dashed, rounded corners=8pt, fill=blue!8,
        inner xsep=3pt, inner ysep=4pt, fit=(D)(d0)(d1)] (badk) {};
\end{scope}
\draw[->, dashed, red!70!black] (E.north) to[bend left=22]
  node[midway, left=1pt, font=\tiny] {climb} (B.south);
\node[red!70!black, font=\scriptsize, align=center] at (4.75, 7.4)
  {$X$ fault at a left child:\\ $\mathrm{bad}=\mathrm{subtree}(\mathrm{parent})$};
\node[blue!70!black, font=\scriptsize, align=center] at (1.7, 2.95)
  {$K_1$ jump at a right child:\\ $\mathrm{bad}=\mathrm{subtree}(v)$};

\node[draw, rounded corners=4pt, fill=gray!6, anchor=north west, minimum width=9.4cm, minimum height=5.2cm] (box) at (6.6, 7.9) {};
\node[font=\small\bfseries, anchor=north west] at (6.85, 7.7) {Prediction of a good branch $(i,j)$ by the XOR mirror};
\node[anchor=north west, align=left, font=\small] at (6.85, 7.15)
  {1. Evolve the reference trajectory, whose data register ends as\\
   \phantom{1. }$\sum_{c'} A(c')\ket{c'}$ with $A(c)=\alpha^{\,k-w}\beta^{\,w}$,\\
   \phantom{1. }$\alpha,\beta=\tfrac12(1\pm a^{2n})$, \; $w=\Ham(c\oplus b_{\mathrm{out}})$;};
\node[anchor=north west, align=left, font=\small] at (6.85, 5.75)
  {2. Map every component by the XOR mirror\\
   \phantom{2. }$c \;=\; c' \oplus \delta_{ij}$, \quad
   $\delta_{ij}= b_{\mathrm{out}}^{\mathrm{ref}} \oplus b_{\mathrm{out}}^{(i,j)}$
   \; $\left(b_{\mathrm{out}}=j\oplus d_{i}\right)$;};
\node[anchor=north west, align=left, font=\small] at (6.85, 4.6)
  {3. Rescale by the address-exposure counter\\
   \phantom{3. }$\mathrm{amp}_{ij}(c) \;=\; \mathrm{amp}_{\mathrm{ref}}(c\oplus\delta_{ij})\,
   (1-\gamma)^{\Delta c_{\mathrm{addr}}/2}$,\\
   \phantom{3. }$\Delta c_{\mathrm{addr}}=c_{\mathrm{addr}}(i)-c_{\mathrm{addr}}(i_{\mathrm{ref}})$
   from the schedule, Eq.~\eqref{eq:caddr}.};
\draw[->, thick, gray] (R.north) to[bend left=18] (6.6, 7.72);
\end{tikzpicture}}
\caption{%
Concept of the pruned simulation of a qubit-encoded $(n,k)$-QRAM for one trajectory, with $n=3$ shown. Each node carries an address qubit, used for routing, and a data qubit. Idle nodes have $a=0$, which is indistinguishable from the pointer value ``route left''. At the top, an $X$-type fault at a left child marks as bad the subtree of its nearest right-child ancestor---the parent in the example drawn. Its stuck excitation climbs up through idle left-pointing ancestors and splits the off-path components into distinct final-configuration families, as captured by Eq.~\eqref{eq:badrange}. A fault at a right child, in contrast, marks only its own subtree. The first good branch, drawn as a thick line, is evolved explicitly together with all bad branches. Every remaining good branch is predicted exactly by the XOR mirror of Eqs.~\eqref{eq:hamming} and~\eqref{eq:goodbranch}. Its data register is the reference's register with the ideal output XOR-shifted, and its amplitude carries the address-exposure factor $(1-\gamma)^{\Delta c_{\mathrm{addr}}/2}$. The marking rule is applied uniformly to every sampled fault, whatever its type, and the prediction is exact on the no-jump side by Theorem~\ref{thm:goodbranch} and on the jump side by the joint sampler of Sec.~\ref{sec:jumps}.}
\label{fig:concept}
\end{figure*}

\section{Introduction}
\label{sec:intro}

Quantum random access memory (QRAM) implements the addressing oracle
\begin{equation}
\sum_{i,j}\alpha_{i,j}\ket{i}_A\ket{j}_D \;\longrightarrow\; \sum_{i,j}\alpha_{i,j}\ket{i}_A\ket{j\oplus d_i}_D
\label{eq:nk-qram}
\end{equation}
with query depth $\mathcal{O}(\log N)$ at the price of $\mathcal{O}(N)$ routing elements~\cite{giovannetti2008qram, giovannetti2008arch}, and underlies quantum search over structured data~\cite{grover1997search}, quantum state preparation~\cite{mottonen2004transformation, wang2025mps}, and quantum machine-learning data loaders~\cite{harrow2009quantum, biamonte2017quantum, kerenidis2017quantum, park2019qram, liu2023datacenter}. Diverse QRAM hardware architectures have been proposed, ranging from hybrid acoustic and spin-photon platforms to superconducting circuits~\cite{hann2019acoustic, kchen2021spin, weiss2023cavity, wang2023hardware}; see Ref.~\cite{jaques2023survey} for a recent survey and critique. Coherent bucket-brigade routing has recently been demonstrated on superconducting processors~\cite{zhang2025router, shen2025experiment, miao2025router}, while fault-tolerant resource estimates~\cite{dimatteo2019fault} and fundamental causal bounds~\cite{ywang2024causal} frame the resource question. Its faithful noisy simulation is therefore an essential tool~\cite{hann2021resilience}, building on general-purpose simulators for open quantum systems and quantum circuits~\cite{johansson2012qutip, jones2019quest, chen2018sim, wang2021sunway}, and was recently made tractable by a sparse-state branch simulator~\cite{wang2025qram} that exploits two observations. The bucket-brigade circuit consists entirely of non-branching reversible gates, so the tree configuration of each input basis component follows a unique classical orbit. Under the Monte Carlo trajectory expansion of the noise channels, each sampled fault affects only those branches whose routing path crosses its subtree. The complementary branches, which we call good, can therefore be predicted analytically instead of evolved. Only the bad branches and one reference good branch are simulated explicitly, giving a cost set by the expected number of bad branches rather than by $D$. Here $D$ is the number of input branches, $n$ the address size, and $p$ the per-element noise rate.

That machinery was developed and implemented for the qutrit encoding, in which each routing node is a three-level system $\{\ket{W},\ket{L},\ket{R}\}$~\cite{zhang2024aux}. Its soundness under amplitude damping, however, has never been analyzed. The simulator applies the no-jump operator $K_0$ deterministically at every time slice, yet the reason a good branch remains predictable under this ever-present attenuation is used implicitly and stated nowhere. The first task of this paper is to write the mechanism down as Proposition~\ref{prop:qutrit}. The operator $K_0$ is diagonal in the node basis, and the wait state $\ket{W}$ is its fixed point, so a good branch follows the noiseless orbit and damping contributes only a scalar factor $a^{c}$ with $a=\sqrt{1-\gamma}$. The exposure count $c$ counts excitations over exposed slices and is computable from the schedule table. The physically simpler qubit encoding, with two two-level systems per node, $\ket{0}\equiv L$, $\ket{1}\equiv R$, and no wait state, lacks both of these properties. Its memory fetch is a controlled phase kickback~\cite{chen2023qram}, which forces the data bus through a Hadamard wall into the $\{\ket+,\ket{-}\}$ basis, through the tree, and back through a second Hadamard wall. Between the two walls, the no-jump operator $K_0=\diag(1,a)$ is not diagonal in the propagation basis, so the scalar predictor fails.

In this paper, we give a complete form of damping-channel predictability for both encodings. We first make the qutrit mechanism explicit, and our main new result, Lemma~\ref{lem:hkh} and Theorems~\ref{thm:multi}--\ref{thm:goodbranch}, is that the qubit $H\!\to\!K_0^d\!\to\!H$ structure of a good branch is exactly solvable: the second Hadamard acts on the tilted state $\left(\ket{0}+(-1)^{b\oplus m}a^{d}\ket{1}\right)/\sqrt{2}$, and the output amplitudes are $(1\pm a^{2n})/2$. We additionally identify a second, previously undocumented qubit-specific effect. $X$-type faults leave stuck excitations that migrate up idle left-pointing ancestors, so the good/bad partition must use a family-divergence rule, under which a fault at a left child marks the subtree of its nearest right-child ancestor and a fault at a right child marks its own; we verify the rule by exhaustive single-fault injection. Combined with the pre-existing address-part scalar counter and this corrected pruning criterion, the result is a complete qubit-encoded QRAM fast-simulation algorithm that evolves only the bad branches plus one reference branch, with a clean error structure: the conditional in-branch infidelity is second order in $\gamma$, and all first-order infidelity is carried by the trajectories that contain sampled faults. We validate the simulator at two levels. In end-to-end benchmarks, the pruned and full modes are driven from identical noise histories and agree seed by seed on the end state, with speedups up to $285\times$. In trajectory-level tests, each closed form and the error structure are verified against explicitly evolved branches and agree to machine precision.

This paper is organized as follows. Section~\ref{sec:background} reviews the bucket-brigade schedule and the noise model. Section~\ref{sec:prunesim} makes the qutrit mechanism explicit (Proposition~\ref{prop:qutrit}) and identifies where the qubit encoding breaks it. Section~\ref{sec:theory} establishes the exactly solvable qubit structure of a good branch (Lemma~\ref{lem:hkh}, Theorems~\ref{thm:multi}--\ref{thm:goodbranch}), defines the exact joint damping sampler (Theorem~\ref{thm:joint}), and analyzes the error structure (Sec.~\ref{sec:errors}). Section~\ref{sec:algorithm} presents the complete qubit-encoded QRAM fast-simulation algorithm, and Section~\ref{sec:results} reports the end-to-end benchmarks and the trajectory-level validation.

\section{Background}
\label{sec:background}

\subsection{Bucket-brigade circuit and time-sliced schedule}

We follow the $(n,k)$-QRAM convention of Refs.~\cite{chen2023qram, wang2025qram}. The setup has $n$ address qubits, $k$ data qubits, classical memory $d_i\in\{0,1\}^k$, and an auxiliary complete binary tree whose nodes at layer $l$ are indexed $v\in[2^{l}-1,\,2^{l+1}-2]$. In the qubit encoding each node carries an address qubit, with $\ket{0}\equiv$ route left and $\ket{1}\equiv$ route right, and a data qubit. There is no idle level. A query proceeds in three phases, namely address setting, data fetching, and uncomputation, using only SWAP and controlled-SWAP type layers plus the CZ phase-kickback fetch, a controlled-$Z$ with the classical memory cell~\cite{barenco1995gates}. The whole schedule is organized into
\begin{equation}
T = 6n + 2k
\end{equation}
mirror steps, of which the time slices $1,\dots,T-1$ carry events. The events relevant to us are, for data qubit $i\in\{0,\dots,k-1\}$,
\begin{align}
\tau_{\mathrm{in}}(i) &= 2n + 2i + 1, \nonumber\\
\tau_{\mathrm{fetch}}(i) &= 3n + 2i + 1, \label{eq:schedule}\\
\tau_{\mathrm{out}}(i) &= 4n + 2i + 1, \nonumber
\end{align}
i.e.\ CopyIn, the swap of the bus qubit into the root data slot, then the CZ fetch at the leaf, and finally CopyOut, with the mirror symmetry $\mathrm{out}(t)=T-t$, which maps the CopyIn time of data bit $i$ to the CopyOut time of bit $k-1-i$. In the qubit architecture the two Hadamard walls act on the whole data-bus register immediately before the first CopyIn and immediately after the last CopyOut. Note that the residence window
\begin{equation}
\tau_{\mathrm{out}}(i)-\tau_{\mathrm{in}}(i) \;=\; 2n
\label{eq:window}
\end{equation}
is the same for every data qubit, independent of the address and of all bit values, a fact that will make the data-bus predictor universal.

\subsection{Noise model and trajectory expansion}
\label{sec:noise}

Every time slice applies its gates and then attacks the tree qubits with noise. The damping layer applies the amplitude damping channel of strength $\gamma$~\cite{breuer2002book},
\begin{align}
K_0&=\ket{0}\bra{0}+\sqrt{1-\gamma}\,\ket{1}\bra{1}=\diag(1,a),\label{eq:kraus}\\
K_1&=\sqrt{\gamma}\,\ket{0}\bra{1},\nonumber
\end{align}
with $a\equiv\sqrt{1-\gamma}$, independently to every one of the $M=2(2^{n}-1)$ tree qubits; on a site carrying $\ket{0}$ in the qubit encoding, or $\ket{W}$ in the qutrit encoding, both Kraus outcomes act trivially. A depolarizing channel of strength $\varepsilon$ applies Pauli $X/Z/Y$ faults for qubits and the eight nontrivial Weyl operators for qutrits, whose shift elements act on all three levels, including the wait state. In the qutrit encoding the damping channel acts on the node levels $\{\ket{W},\ket{L},\ket{R}\}$ as~\cite{zhang2024aux, wang2025qram}
\begin{align}
K_0^{\mathrm{tri}}&=\ket{W}\bra{W}+a\left(\ket{L}\bra{L}+\ket{R}\bra{R}\right),\label{eq:kraus-tri}\\
K_{1}^{\mathrm{tri}}&=\sqrt{\gamma}\ket{W}\bra{L},\qquad
K_{2}^{\mathrm{tri}}=\sqrt{\gamma}\ket{W}\bra{R}.\nonumber
\end{align}
Two features of Eq.~\eqref{eq:kraus-tri} will be decisive. The operator $K_0^{\mathrm{tri}}$ is diagonal in the node basis, and the wait state is its fixed point, $K_0^{\mathrm{tri}}\ket{W}=\ket{W}$. Following Ref.~\cite{wang2025qram}, a simulation shot samples the complete noise history in advance, the quantum-trajectory Monte Carlo strategy~\cite{dalibard1992monte, plenio1998quantum} previously applied to quantum-algorithm noise benchmarking~\cite{xue2021noise}. Depolarizing faults apply their sampled unitary at their sampled sites. A damping layer is resolved as one joint Kraus outcome over the whole tree: a candidate set of sites is drawn state-independently, one auxiliary computational-basis draw from the current state selects the joint jump set among the candidates, and the product Kraus operator realizing that jump pattern is applied to the coherent state, followed by a single normalization. Section~\ref{sec:jumps} defines this joint sampler and proves that the trajectory average reproduces the trace-preserving channel exactly. Between layers the state is kept as an unnormalized sparse pure state; the norm is restored once per damping layer and at the final sampling. Averaging over shots reproduces channel-level observables.

\section{The pruning simulation of QRAM}
\label{sec:prunesim}
\subsection{Qutrit-encoded QRAM}
\label{sec:structure}

The pruning simulator of Ref.~\cite{wang2025qram} treats the QRAM as a program in which no gate ever branches. For each address $\ket{i}$ the tree-configuration orbit $\ket{\Psi_i(t)}$ is unique, and therefore classically computable without evolving the register. Its noise inherits the same rigidity. $K_0$ is diagonal in the node computational basis, so a good branch follows the noiseless orbit and damping contributes only the scalar $a^{c}$. The exposure count $c$ counts excitations over exposed slices and is read off the schedule table. The relative weight of two good branches is then $a^{\Delta c}$. Faults, meanwhile, stay where they occur. A fault at node $(l,p)$ affects only addresses in the leaf range of its subtree,
\begin{equation}
i\in\left[\,2^{\,n-l}p,\;2^{\,n-l}(p+1)-1\,\right],
\label{eq:subtree}
\end{equation}
so the good/bad partition is determined before evolution, with expected bad fraction $\mathcal{O}(n^{2}p)$. Appendix~\ref{app:prelim} summarizes the simulator's data representation, error sampling, and pruned evolution.

For qutrits the criterion rests on a stronger structural fact that is easy to overlook. A fault leaves every off-path component with one and the same final tree configuration. An address-slot flip of the flip type acts inside the $\{\ket{L},\ket{R}\}$ logic levels and cannot touch a wait-level qutrit; a stray excitation in a data slot is frozen in place by the $\ket{W}$ guard of the routing layer; and a shift-type fault, which does reach $\ket{W}$, relocates the affected idle site to a definite level inside its own subtree, identically in every off-path component. No excitation ever crosses from one off-path component's region into another's. As Sec.~\ref{sec:gap} shows, what fails for qubits is this configuration uniformity. Output damage containment, by contrast, holds in both encodings.

The diagonal damping assumption above is used in Ref.~\cite{wang2025qram} but never stated, and it is the assumption this paper generalizes. For the qutrit encoding it covers the data bus automatically. Qutrit fetching is a computational-basis CNOT, so no Hadamard walls exist and a good branch never splits.

\begin{proposition}[Qutrit damping predictability]
\label{prop:qutrit}
Let branch $(i,j)$ of the qutrit-encoded QRAM be good for a sampled noise history, i.e.\ its routing path avoids every sampled fault. Then its no-jump trajectory end state is
\begin{equation}
\ket{\Psi^{\good}_{i,j}}
=\;a^{\,c(i,j)}\,\ket{i}_A\ket{j\oplus d_i}_D\ket{Q_0}_{\mathrm{tree}},
\label{eq:qutritgood}
\end{equation}
where $\ket{Q_0}$ is a tree configuration shared by all good branches, fixed by the sampled history and equal to the idle all-$\ket{W}$ configuration when no sampled fault touches an idle site of a good component, and $c(i,j)\in\mathbb{N}$ is an exposure count, namely the excited $L/R$-valued slots summed over time slices, computable from the schedule table without evolution. The relative weight of two good branches is $a^{\Delta c}$ with $\Delta c=c(i,j)-c(i',j')$, and the in-branch deviation from the ideal output is exactly zero.
\end{proposition}

\begin{proof}
Since the branch is good, no sampled fault lies on its path, so only the no-jump operator $K_0^{\mathrm{tri}}$ of Eq.~\eqref{eq:kraus-tri} acts on its support. Being diagonal, $K_0^{\mathrm{tri}}$ cannot alter the unique basis orbit, which the non-branching schedule fixes in advance. It multiplies the branch amplitude by one factor of $a$ per excited slot per slice, while slots in the wait state are unattenuated. Since the fetch is a computational-basis CNOT, the data bus never leaves the computational basis and no wall-induced splitting occurs, so the data register ends in the definite word $j\oplus d_i$. Summing the per-slice excitation counts yields the schedule quantity $c(i,j)$. After uncomputation the tree returns to a configuration determined by the sampled faults; a fault off the routing path acts identically on every off-path component, so the configuration is one and the same for all good branches, and it is the idle all-$\ket{W}$ configuration when no sampled fault touches an idle site. This proves Eq.~\eqref{eq:qutritgood}. After normalization the branch coincides with the ideal output exactly.
\end{proof}

Equation~\eqref{eq:qutritgood} says for the damping what Ref.~\cite{wang2025qram} already uses for reliable branches. Every good branch ends in one and the same tree configuration $\ket{Q_0}$, and differs from the ideal output only by the scalar $a^{c}$. We call this the \emph{shared-final-state property} of reliable branches. The end state of a good branch is its ideal output times one common tree configuration, and only the scalar remains branch-dependent. The exponent $c$ does depend on the branch, since a data-slot excitation is present on a given leg only for particular input and output bit values. It is fixed by the schedule alone, so it can be read off without evolving the register. That is all the pruning algorithm needs.

\subsection{Qubit-encoded QRAM}
\label{sec:gap}

Proposition~\ref{prop:qutrit} rests on three simple properties of the qutrit circuit. The gates never branch, so every branch follows one fixed orbit through the tree. The damping operator $K_0$ is diagonal in the basis in which each traveling qubit propagates. Idle nodes sit in the wait state $\ket{W}$, which the flip-type faults cannot touch and which the shift-type faults disturb identically in every off-path component. In the qubit encoding the orbit is still unique, and damping of the address qubits is still a computable scalar. Routing qubits excited to $\ket{1}$ are damped at every slice they persist, with windows read off the schedule. The relative weight between good branches is again $a^{2\Delta c_{\mathrm{addr}}}$, with the schedule-computable integer $\Delta c_{\mathrm{addr}}$ given in Eq.~\eqref{eq:caddr}.

The wait-state protection is what fails first, under $X$-type depolarizing faults. A qubit node has no wait level. A pointer value of $0$, meaning route left, and an idle node are therefore indistinguishable, and every node's controlled swap fires. An $X$ flip therefore leaves a permanent stuck excitation in every superposition component, including components in which the fault node is off-path, where nothing ever cleans it. The stuck excitation then migrates. An idle parent always points left, so a stray excitation sitting in a left child is pulled up into the parent's data slot whenever the parent's swap fires. A parent holding $0$ thereby propagates the stray. A stray in a right child is never touched by an idle parent, since only a parent holding $1$ swaps with its right child. It is instead pushed down the leftmost chain of its own subtree. Output damage itself remains confined to the fault node's subtree in both encodings. What differs is the final configuration of off-path components. Components whose strays stop at different heights form distinct configuration families, and the end-of-trajectory measurement of the tree couples them. No single family can be predicted while another survives. The family-divergence envelope, illustrated in Fig.~\ref{fig:concept}, is exactly
\begin{equation}
\mathrm{bad}(v)=
\begin{cases}
\mathrm{subtree}(v), & \text{right child or root's child},\\[2pt]
\mathrm{bad}(\mathrm{parent}(v)), & \text{left child},
\end{cases}
\label{eq:badrange}
\end{equation}
with $\mathrm{bad}(\mathrm{root})$ the whole leaf range. The recursion is the climb itself: a stray deposited at a left child is pulled up through the chain of idle left-pointing ancestors until it reaches a right-child ancestor, and that ancestor's subtree contains every position the stray can come to rest in. The root always holds the first address bit, so a stray at its left child is never pulled up for off-path components, and the root's children serve as terminal cases. Branches inside the marked range need not be damaged; they have merely lost the guarantee that their final tree configuration coincides with the reference branch's, so the simulator evolves them explicitly instead of predicting them. We verified Eq.~\eqref{eq:badrange} exhaustively by single-fault injection on an $n=3$ tree, against a full-evolution reference simulator, covering all $644$ combinations of fault position, time slice, and fault type with zero violations. The details are described in Sec.~\ref{sec:results}. The injection covers single faults; for several faults the rule marks the union of the per-fault ranges, and the multi-fault consequences are covered by the trajectory-level tests of Sec.~\ref{sec:results}, though no general multi-fault theorem is claimed here. For every fault position, slot, and time slice, the wrong-bus set is contained in $\mathrm{subtree}(v)$ and the configuration families diverge precisely up to the boundary of Eq.~\eqref{eq:badrange}. The qutrit circuit under the same injections keeps every off-path component configuration-identical. The rule is applied uniformly to every sampled fault, whatever its type. Amplitude damping cannot create strays of its own: a jump requires an excitation, and idle qubits have none, so the damage of a pure damping fault stays inside its own subtree, and for such faults the uniform rule marks a superset. The extra marked branches are simply evolved explicitly, which keeps the criterion type-independent and the pruned evolution exact. The data-bus results of Sec.~\ref{sec:theory} hold for every branch whose path avoids the sampled faults.

What finally fails, however, is the diagonal damping of the data bus itself. Between the two Hadamard walls every data qubit propagates as
\begin{equation}
\ket{\pm}=\tfrac{1}{\sqrt2}\left(\ket{0}\pm\ket{1}\right),
\qquad
K_0\ket{\pm}=\tfrac{1}{\sqrt2}\left(\ket{0}\pm a\ket{1}\right),
\end{equation}
which leaves the propagation axis. The two components of the traveling qubit acquire different attenuations, the two walls no longer cancel, and the output of the second wall is not a definite computational value. A scalar predictor cannot represent this, and a naive simulator must split and re-merge the branch at every wall.

\begin{figure}[t]
\centering
\includegraphics[width=\columnwidth]{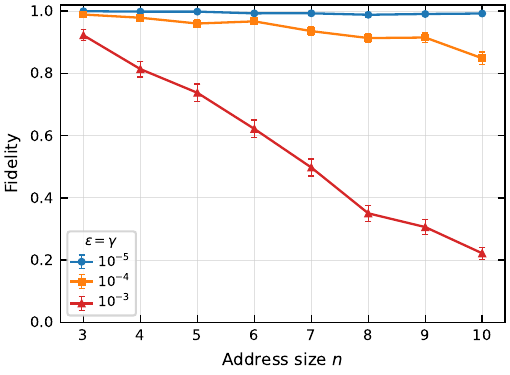}
\caption{Trajectory-averaged fidelity of the qubit encoding versus address size $n$, for $\varepsilon=\gamma$ from $10^{-5}$ to $10^{-3}$, with $200$ trajectories per point, $k=3$, and $500$ input branches; error bars show the standard error of the mean over the trajectories.}
\label{fig:fidelity}
\end{figure}

\section{Closed-form prediction of the $H\to K_0^d\to H$ structure}
\label{sec:theory}

\subsection{Single-qubit lemma}

\begin{lemma}[$H$--$K_0$--$H$ lemma]
\label{lem:hkh}
Let $a=\sqrt{1-\gamma}$, $d\in\mathbb{N}$, $m\in\{0,1\}$, and
\begin{equation}
N_{m,d} \;=\; H\,Z^{m}\diag(1,a^{d})\,H .
\end{equation}
Then for every input bit $b\in\{0,1\}$, writing $b_{\mathrm{out}}=b\oplus m$,
\begin{align}
N_{m,d}\ket{b} &= \alpha_d\,\ket{b_{\mathrm{out}}} + \beta_d\,\ket{\overline{b}_{\mathrm{out}}},
\label{eq:lemma}\\
\alpha_d&=\frac{1+a^{d}}{2},\qquad \beta_d=\frac{1-a^{d}}{2}.\nonumber
\end{align}
\end{lemma}

\begin{proof}
$H\ket{b}=(\ket{0}+(-1)^{b}\ket{1})/\sqrt2$. Between the walls the data qubit experiences only $d$ applications of $K_0$, which merely relocate it through the routing SWAPs, and the CZ fetch, a $Z^m$ phase kick. Both are diagonal in the computational basis and commute, so they merge into $Z^m\diag(1,a^d)=\diag(1,(-1)^m a^d)$. The state before the second wall is $\left(\ket{0}+(-1)^{b\oplus m}a^{d}\ket{1}\right)/\sqrt2$. Applying $H$ gives $\tfrac12\left[(1+(-1)^{b\oplus m}a^{d})\ket{0}+(1-(-1)^{b\oplus m}a^{d})\ket{1}\right]$, and collecting by $b_{\mathrm{out}}=b\oplus m$ yields Eq.~\eqref{eq:lemma}.
\end{proof}

The state in Eq.~\eqref{eq:lemma} is not normalized, and the missing weight has an exact meaning. Between the two walls the qubit is in the state $(\ket{0}+(-1)^{b\oplus m}a^{d}\ket{1})/\sqrt{2}$, whose squared norm $\alpha_d^2+\beta_d^2=(1+a^{2d})/2$ is precisely the probability of crossing all $d$ damping layers without a jump. The complementary weight $(1-a^{2d})/2$ is the probability that at least one jump occurs, and it is carried by the bad branches, which the simulator evolves explicitly. No probability is lost. The norm of the no-jump branch plus the weight of the branches with jumps always sums to one. The limiting cases check out as well. At $a=1$, i.e.\ no damping, $\beta_d$ vanishes and the fetch is ideal, the two Hadamard walls cancelling since $H^2=\mathds{1}$. At $d=0$ the qubit crosses no damping layer, $\diag(1,a^{0})=\mathds{1}$, and the map again reduces to the ideal fetch.

\subsection{Multi-qubit product structure}

\begin{theorem}[$k$-qubit output formula]
\label{thm:multi}
Consider a good branch with address $i$, bus input $j\in\{0,1\}^k$, stored word $d=d_i$, ideal output $b_{\mathrm{out}}=j\oplus d$; for a bit string $x$, $\Ham(x)$ denotes its Hamming weight. In the qubit encoding each data qubit resides in the tree for exactly the $2n$ damping layers of Eq.~\eqref{eq:window}, independently of everything else, and different data qubits share no gates. Hence the data register after the second wall is
\begin{equation}
\begin{aligned}
\bigotimes_{t=0}^{k-1} N_{d^{(t)},\,2n}\,\ket{j_t}
&=\sum_{c\in\{0,1\}^k} A(c)\,\ket{c},\\
A(c)&=\alpha_{2n}^{\,k-w}\,\beta_{2n}^{\,w},\\
w&=\Ham\!\left(c\oplus b_{\mathrm{out}}\right),
\end{aligned}
\label{eq:hamming}
\end{equation}
i.e.\ the amplitude of an output word $c$ depends only on the Hamming weight $w$ of its error pattern relative to the ideal output.
\end{theorem}

\begin{proof}
The two walls bracket a product evolution. Routing layers relocate qubits but never couple two data qubits, the $Z^{d^{(t)}}$ kicks act on distinct qubits, and the no-jump operator $K_0^{\otimes}$ factors into per-qubit diagonal maps, each acting on one factor of a product state. Apply Lemma~\ref{lem:hkh} per qubit and collect output words by the number of flipped bits.
\end{proof}

Note the contrast with the qutrit scheme, where the branch dependence of the exposure window forces the per-branch count $c(i,j)$ of Proposition~\ref{prop:qutrit}. In the qubit scheme the $X$-basis transport makes the window universal, equal to $2n$ always, which simplifies rather than complicates the counting.

\subsection{Address part and the good-branch state}

The routing excitations persist for the whole query in the qubit encoding, since there is no wait level to relax into. The schedule gives the exposure count
\begin{align}
c_{\mathrm{addr}}(i)&=\sum_{t\,:\,i_t=1}\left(T-2(2t+1)\right)\nonumber\\
&\;=\;\sum_{t\,:\,i_t=1}\left(6n+2k-4t-2\right),
\label{eq:caddr}
\end{align}
so that the relative probability weight of good branch $i$ with respect to the reference branch $i_{\mathrm{ref}}$ is
\begin{equation}
\frac{p_i}{p_{i_{\mathrm{ref}}}}=(1-\gamma)^{\,c_{\mathrm{addr}}(i)-c_{\mathrm{addr}}(i_{\mathrm{ref}})} .
\label{eq:relmult}
\end{equation}

\begin{theorem}[good-branch closed form]
\label{thm:goodbranch}
For a sampled noise history, let branch $(i,j)$ be good, i.e.\ its routing path avoids every sampled fault. Then its no-jump trajectory end state is
\begin{equation}
\ket{\Psi^{\good}_{i,j}}
=\;a^{\,c_{\mathrm{addr}}(i)}\,\ket{i}_A\otimes
\left[\bigotimes_{t=0}^{k-1}N_{d^{(t)},2n}\ket{j_t}\right]_D\otimes
\ket{\tilde Q_0}_{\mathrm{tree}},
\label{eq:goodbranch}
\end{equation}
where $\ket{\tilde Q_0}$ is the unnormalized remnant of the tree state, identical for all good branches.
\end{theorem}

\begin{proof}
The reversible orbit is unique because the circuit never branches, and no sampled fault touches the path of a good branch, so between slices only $K_0^{\otimes}$ acts. Splitting $K_0^{\otimes}$ by qubit, routing qubits contribute the scalar of Eq.~\eqref{eq:caddr}, each data qubit contributes the wall-to-wall map of Lemma~\ref{lem:hkh}, and the tree returns through uncomputation to the common configuration $\ket{\tilde Q_0}$.
\end{proof}

Theorem~\ref{thm:goodbranch} generalizes Proposition~\ref{prop:qutrit}, the shared-final-state property of reliable qutrit branches, to the qubit encoding. The sharing survives. Only the data register is promoted from a definite word to the known two-component product state of Theorem~\ref{thm:multi}. The scalar $a^{c}$ is promoted to the address scalar together with the Hamming formula.

\begin{table*}[t]
\centering
\caption{Qutrit vs.\ qubit encoding in the pruning simulator.}
\label{tab:compare}
\begin{tabular}{lcc}
\toprule
Aspect & Qutrit & Qubit, this work \\
\midrule
Idle level & $\ket{W}$ fixed point of $K_0$ & None; $\ket{1}$ routers persist \\
Memory fetch & Computational-basis CNOT & CZ phase kickback, $H$ walls \\
Walls & Degenerate, merge only & Real split/merge \\
Data-qubit exposure & Basis-valued windows & Universal $2n$, value-independent \\
Good-branch data output & Definite $\ket{j\oplus d_i}$ & Known product, Eq.~\eqref{eq:hamming} \\
In-branch error & $0$, pure scalar & $\mathcal{O}(n^2\gamma^2)$ leakage / qubit \\
Address exposure & $\mathcal{O}(t)$ per bit, transient & $\sim T$ per bit, persistent \\
$X$-fault residue & Frozen at fault node by the $W$ guard & Migrates up idle ancestors \\
Bad range & $\mathrm{subtree}(v)$ & Eq.~\eqref{eq:badrange}, left child $\to$ nearest right-child ancestor \\
Predictor & Scalar $a^{\Delta c}$ & Scalar $\times$ Hamming formula \\
\bottomrule
\end{tabular}
\end{table*}

Table~\ref{tab:compare} summarizes the structural origins of the contrast between the two encodings, aspect by aspect. They include the presence of a wait level, the basis of the memory fetch, the persistence of routing excitations, and the pruning criteria that follow. The in-branch error row is quantified in Sec.~\ref{sec:errors}.

\subsection{Exact joint sampling of damping}
\label{sec:jumps}

Theorem~\ref{thm:goodbranch} makes the replacement exact on the no-jump side. The jump side needs a sampler whose outcome distribution is exact on the full multi-qubit state. This subsection defines one and proves its exactness, then shows that the pruned mode evaluates it from the explicitly represented branches and the predicted scalars alone.

Consider one damping layer, acting on the $M=2(2^{n}-1)$ tree qubits of Sec.~\ref{sec:noise}, and write its input as $\ket{\psi}=\sum_x c_x\ket{x}$ in the computational basis, where $x$ records the tree sites together with all spectator registers; spectator labels are marginalized automatically. Let $X(x)$ denote the set of tree sites excited in configuration $x$. For a set $J$ of sites the layer Kraus operator is
\begin{equation}
K_J=\prod_{q\in J}K_{1,q}\prod_{q\notin J}K_{0,q},
\label{eq:layerkraus}
\end{equation}
with both products over the tree sites; operators on different sites commute. A shot resolves the layer in four steps. First, draw a candidate set $C$, including each tree site independently with probability $\gamma$; this draw is state-independent and can be made in advance for all layers. Second, draw one auxiliary configuration $x$ with probability $|c_x|^{2}$ from the current state. Third, set the joint jump set to $J=C\cap X(x)$. Fourth, apply $K_J$ to the original coherent state $\ket{\psi}$ and normalize the conditional state. The auxiliary draw only selects the environment outcome; the simulated state is never replaced by $\ket{x}$.

\begin{theorem}[Exact joint damping sampling]
\label{thm:joint}
For independent amplitude damping over the domain of Sec.~\ref{sec:noise}, the sampler above produces the exact Kraus-outcome distribution $p_J=\bra{\psi}K_J^{\dagger}K_J\ket{\psi}$ and the exact channel average, for every normalized pure input and every $\gamma<1$.
\end{theorem}

\begin{proof}
Conditioned on the auxiliary draw $x$, only sites in $X(x)$ can enter $J$, and the candidate choices at the complementary sites cancel from the normalization, giving
\begin{equation}
\Pr(J\mid x)=\mathbf{1}_{J\subseteq X(x)}\,\gamma^{|J|}(1-\gamma)^{|X(x)|-|J|}.
\end{equation}
Marginalizing over the auxiliary draw, and using that $K_J^{\dagger}K_J$ is diagonal in the computational basis with the same factor as its value on $\ket{x}$,
\begin{align}
\Pr(J)&=\sum_{x:\,J\subseteq X(x)}|c_x|^{2}\,\gamma^{|J|}(1-\gamma)^{|X(x)|-|J|}\nonumber\\
&=\bra{\psi}K_J^{\dagger}K_J\ket{\psi}.
\end{align}
The conditional state of an outcome with $p_J>0$ is $K_J\ket{\psi}/\sqrt{p_J}$, so the shot average is
\begin{equation}
\sum_{J:\,p_J>0}p_J\,\frac{K_J\ket{\psi}\bra{\psi}K_J^{\dagger}}{p_J}
=\sum_{J}K_J\ket{\psi}\bra{\psi}K_J^{\dagger},
\end{equation}
which is the channel of Eq.~\eqref{eq:kraus} applied independently to every tree site. The identity holds at finite $\gamma$ and involves neither a single-jump approximation nor a first-order time discretization.
\end{proof}

The pruned mode must draw the auxiliary configuration from the complete state, including the predicted branches: evaluated on the explicitly represented branches alone, the jump statistics would miss the predicted weight and the trajectory ensemble would be biased. The candidate set is state-independent and shared, so what this summation must supply is the mask weight
\begin{equation}
w(J)=\sum_{x:\,C\cap X(x)=J}|c_x|^{2}
\label{eq:maskweight}
\end{equation}
of the whole represented state; one weighted draw over these weights is equivalent to drawing $x$ and intersecting it with $C$. The required object is therefore the joint distribution of masks across all candidate sites, not merely the one-site excited populations. It is available without evolution. By Theorem~\ref{thm:goodbranch}, the components of every predicted good group are in one-to-one correspondence with the reference branch's components: the mirrored component's tree configuration follows the branch's own deterministic orbit, and its amplitude is the reference component's amplitude times known analytic factors, the address-exposure ratio of Eq.~\eqref{eq:relmult} and the wall-to-wall data factors. Its excitation set, and hence its contribution to $w(J)$, can therefore be evaluated directly. Summing these contributions over the predicted groups, on top of the explicitly represented bad branches and the reference, reproduces the mask-weight map that the full mode would compute. One auxiliary weighted draw over the summed weights then selects the same joint outcome in both modes, and $K_J$ is applied to the represented components. If the reference branch carries no surviving component, its mask distribution is empty and the predicted groups contribute zero weight, in agreement with the reconstruction rule of Sec.~\ref{sec:algorithm}. The joint agreement of the summed weights with the full-mode map is verified by the multi-candidate regression tests of Sec.~\ref{sec:results}.

The sampler, like the pruning algorithm as a whole, resolves the input into branches labeled by an address and one computational-basis bus word. Branches that share an address but differ in the bus word act as orthogonal input columns: their probabilities enter the weights and the outcome statistics, and the bookkeeping does not track interference between different columns. The established input class is thus a coherent superposition over addresses with one bus-input column per address, which includes the data-loading states used throughout, and arbitrary superpositions within a single branch are covered by the local-channel tests of Sec.~\ref{sec:results}. General coherent superpositions across several bus-input columns of one address require a coherent sum of the corresponding amplitudes and are outside the scope of the present simulator.

\subsection{Norm loss versus coherent leakage}
\label{sec:errors}

The closed form serves a dual purpose. It not only underlies the efficient simulation algorithm of Sec.~\ref{sec:algorithm} but also explains how the error of the qubit encoding differs from that of the qutrit encoding. Here, we determine which deviations of the predicted state are genuine errors, and at which order. Expanding per data qubit with $d=2n$,
\begin{align}
\alpha_{2n}&=\frac{1+(1-\gamma)^{n}}{2}=1-\frac{n\gamma}{2}+\mathcal{O}(\gamma^2),\nonumber\\
\beta_{2n}&=\frac{1-(1-\gamma)^{n}}{2}=\frac{n\gamma}{2}+\mathcal{O}(\gamma^2).
\end{align}
Two notions must be separated, and the first-order expansion above is what makes the separation visible. The predicted state enters the error analysis twice, with different physical meaning. Its norm $\alpha^2+\beta^2$ measures the probability weight handed over to the jump trajectories. Its shape after normalization, $\beta^2/(\alpha^2+\beta^2)$, measures the distortion remaining within a no-jump trajectory. The first quantity is first order in $\gamma$ but is not an error of the prediction. The second is the true in-branch error and is second order.

The first notion is norm loss, which is first order and not an error. The no-jump probability per data qubit is $\alpha^2+\beta^2=(1+a^{4n})/2=1-n\gamma+\mathcal{O}(\gamma^2)$. The lost weight is the probability mass of the jump trajectories, carried by the explicitly simulated bad branches. The shrinking unnormalized norm of a good branch is not an output error.

The second notion is the in-branch distortion, which we call coherent leakage: it is second order in the conditional infidelity, and it involves no escape from the computational subspace. Conditional on no jump, the per-qubit error probability is
\begin{equation}
\frac{\beta_{2n}^2}{\alpha_{2n}^2+\beta_{2n}^2}
=\frac{n^2\gamma^2}{4}+\mathcal{O}(\gamma^3).
\end{equation}
Hence the conditional infidelity of a good branch---the infidelity of its normalized no-jump state against the ideal output---is second order in $\gamma$, while the trace distance of the same states, $\sqrt{q}=\mathcal{O}(n\gamma)$, remains first order: the error amplitude is first order and only its probability is second order. All first-order infidelity of the channel average lives on the jump trajectories. This is consistent with the channel-averaged single-qubit error $\frac{1-a^{2n}}{2}$, which decomposes exactly as
\begin{equation}
\frac{1-a^{2n}}{2}=\frac{1-a^{4n}}{4}+\frac{(1-a^{2n})^{2}}{4},
\end{equation}
a first-order jump contribution plus a second-order no-jump distortion, the latter being the coherent leakage above. For error-filtration applications this means that postselecting or conditioning on no-jump trajectories leaves a conditional infidelity of $\mathcal{O}(n^2\gamma^2)$ per data qubit.

\begin{figure*}[t]
\centering
\includegraphics[width=\textwidth]{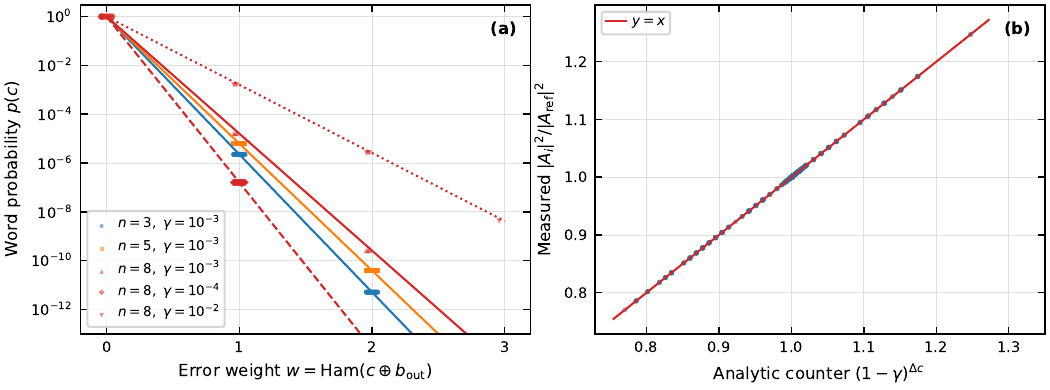}
\caption{Targeted verification of the closed-form predictions under damping. (a)~Normalized output-word probabilities of the data register in no-jump trajectories of the qubit encoding, shown as points, against Eq.~\eqref{eq:hamming} shown as lines. Words are grouped by the Hamming weight $w=\Ham(c\oplus b_{\mathrm{out}})$ of the error pattern, and words of the same $w$ collapse onto one value. All words with theoretical probability down to $10^{-12}$ agree to floating-point accuracy. (b)~Qutrit encoding. Measured relative branch weights $|A_i|^2/|A_{\mathrm{ref}}|^2$ versus the analytic exposure counter $(1-\gamma)^{\Delta c}$ of Proposition~\ref{prop:qutrit}. The measured weights agree with the analytic counter to machine precision, and every good branch ends in the ideal output.}
\label{fig:verify_hamming}
\end{figure*}

\section{Algorithm and complexity}
\label{sec:algorithm}

\begin{algorithm}[t]
\caption{Pruned simulation of the qubit-encoded $(n,k)$-QRAM for one trajectory.}
\label{alg:prune}
\begin{algorithmic}[1]
\Require sparse state $\sum_{i,j} c_{i,j}\ket{i}\ket{j}$, memory $d$, noise $(\varepsilon,\gamma)$, $a=\sqrt{1-\gamma}$
\State sample the noise history over slices $1..T-1$ with $T=6n+2k$; mark each fault, whatever its type, into the bad-address union $\mathcal{R}$ by the configuration-divergence rule of Eq.~\eqref{eq:badrange}, which maps a left child to the subtree of its nearest right-child ancestor and a right child to its own subtree
\State partition: $i\in\mathcal{R}\Rightarrow\bad$; keep the first good branch as reference
\State evolve explicitly the bad branches $+$ the reference, with $H$ walls splitting and merging the bus, and one joint damping layer per slice as in Sec.~\ref{sec:jumps}
\State for each good branch $(i,j)$: address factor $r_i=(1-\gamma)^{c_{\mathrm{addr}}(i)-c_{\mathrm{addr}}(i_{\mathrm{ref}})}$ of Eq.~\eqref{eq:caddr}; data factors $\alpha=(1+a^{2n})/2$, $\beta=(1-a^{2n})/2$, $b_{\mathrm{out}}=j\oplus d_i$
\State emit components $c\in\{0,1\}^k$ with amplitude $c_{i,j}\sqrt{r_i}\,\alpha^{k-w}\beta^{w}$, $w=\Ham(c\oplus b_{\mathrm{out}})$
\State reconstruct the output sparse state; sample and normalize
\end{algorithmic}
\end{algorithm}

Algorithm~\ref{alg:prune} summarizes the complete simulator, and Fig.~\ref{fig:concept} illustrates its central idea. The noise history is sampled in advance. Each sampled fault, whatever its type, marks the bad region given by the family-divergence range of Eq.~\eqref{eq:badrange}. The bad branches and one reference good branch are evolved explicitly. Every remaining good branch is reconstructed from the reference branch without further evolution, with exactness guaranteed by Theorem~\ref{thm:goodbranch} of Sec.~\ref{sec:theory} on the no-jump side and by the joint sampler of Sec.~\ref{sec:jumps} on the jump side.

The expected number of bad branches follows from the size of the marking rule. A branch is marked once a fault falls at any node $v$ whose range $\mathrm{bad}(v)$ contains its address, so the marking sources of one address are the nodes $v$ with $i\in\mathrm{bad}(v)$. Under the recursion of Eq.~\eqref{eq:badrange}, a right child marks its own subtree while every left child inherits the range of its nearest right-child ancestor, and summing the range sizes over the tree gives the average count of marking sources per address,
\begin{equation}
\sum_{v}\frac{|\mathrm{bad}(v)|}{2^{n}}\;=\;n+\frac{(n-1)(n-2)}{4}\;=\;\Theta(n^{2}),
\label{eq:marksources}
\end{equation}
which replaces the $\mathcal{O}(n)$ count of the plain subtree rule. Each candidate node can be sampled faulty at any of the $T=\mathcal{O}(n)$ time slices, with per-element probability $p=\mathcal{O}(\gamma+\varepsilon)$, the rate of the noise channels of Sec.~\ref{sec:noise}. With $D$ input branches, the expected number of bad branches is therefore
\begin{equation}
\mathbb{E}[B]=\mathcal{O}(n^{3}p)\,D,
\label{eq:eb}
\end{equation}
where $B$ denotes the number of bad branches.

Table~\ref{tab:cost} collects the resulting per-branch costs and the savings of the pruned mode. The factor $2^{k}$ in the prediction cost is the width of the damped data register itself. By Eq.~\eqref{eq:hamming}, each good branch ends in a superposition over all $2^{k}$ output words, with amplitudes fixed by the Hamming weight alone. This width is physical rather than algorithmic, and the simulator never needs to store it. Any single amplitude can be evaluated on demand from Eq.~\eqref{eq:hamming}. If only low-weight error patterns are of interest, the sum can be truncated at weight $w_{\max}$ with residual error $\mathcal{O}((n\gamma)^{w_{\max}+1})$, since each unit of Hamming weight costs a factor $\beta/\alpha=\mathcal{O}(n\gamma)$ in amplitude. The remaining costs---sampling the noise history over the $2(2^{n}-1)$ tree sites, storing the classical memory, and producing the requested output---are common to the pruned and the full mode and are inherited from the reference simulator unchanged. The advantage of the pruning is therefore a single, cleanly separable replacement: $D$ explicit branch evolutions become $\mathbb{E}[B]+1$, and the resulting speedup is measured directly by the pruned-versus-full comparison of Sec.~\ref{sec:results}.

\begin{table*}[t]
\centering
\caption{Per-good-branch cost and the savings of the pruned mode.}
\label{tab:cost}
\begin{tabular}{lcc}
\toprule
 & Qutrit & Qubit, this work \\
\midrule
Explicitly evolved branches & $B+1$ (full: $D$) & $B+1$ (full: $D$) \\
Prediction cost / branch & $\mathcal{O}(n+k)$ & $\mathcal{O}(n+k)+\mathcal{O}(k\,2^{k})$ \\
Memory / branch & $\mathcal{O}(n)$ & $\mathcal{O}(n+k)$, product on demand \\
Shared by both modes & \multicolumn{2}{c}{noise sampling over the tree, classical memory, output} \\
\bottomrule
\end{tabular}
\end{table*}

\section{Numerical experiments}
\label{sec:results}

\begin{figure*}[t]
\centering
\includegraphics[width=\textwidth]{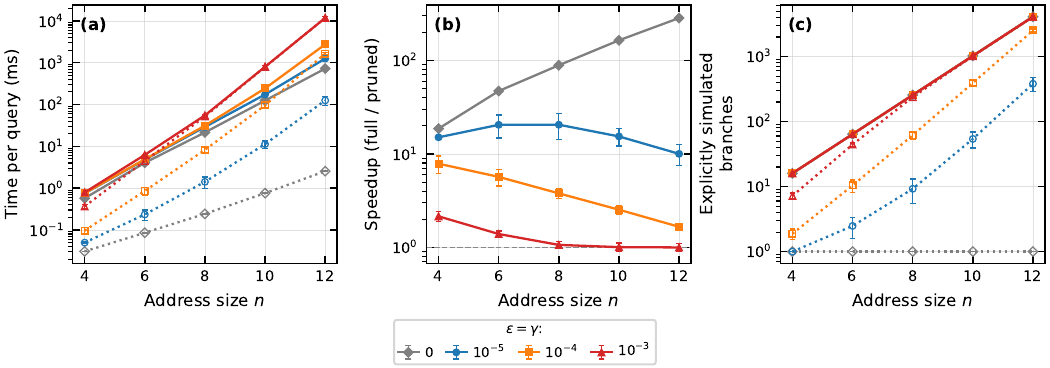}
\caption{Full versus pruned simulation of the qubit encoding with $k=3$ under the data-loading input of one branch per address with zero bus input, so that the full mode evolves exactly $D=2^n$ branches, averaged over $50$ trajectories. Error bars show the standard error of the mean over the trajectories. (a)~Wall-clock time per query. (b)~Pruned-over-full speedup, the ratio of the mean wall-clock times of the two modes. (c)~Explicitly simulated branches. All input branches in the full mode, shown solid, are compared with the bad branches plus the single reference branch in the pruned mode, shown dotted; the speedup in panel~(b) is the ratio of the two. At zero noise the pruned mode evolves exactly one branch, and the speedup rises to $285\times$ at $n=12$. At $\varepsilon=\gamma=10^{-5}$ the pruned mode evolves $3.6\%$--$9.3\%$ of the branches at speedups of $10$--$21\times$; at $10^{-3}$ almost every branch is bad and the speedup approaches unity.}
\label{fig:perf}
\end{figure*}

The algorithm is implemented in the sparse-state QRAM simulator framework~\cite{sparqsim2025}, with a branch-level qubit circuit whose pruned and full modes are driven by the same noise history. Whereas generic state-vector and trajectory simulators evolve the full register in every trajectory~\cite{johansson2012qutip, jones2019quest, chen2018sim, wang2021sunway}, this framework exploits the branch structure of the QRAM circuit, and the pruning layer removes even the good branches from explicit evolution. The implementation is validated as follows. For a fixed random seed, the two modes share the identical sampled noise history, jump decisions, and final collapse, so any discrepancy in the end-state fidelity exposes an implementation or prediction error. All runs use $k=3$ data bits. The fidelity and verification scans use an input of $500$ branches sampled without replacement from the $2^{n+k}$ address--bus components with uniform weights; at the smallest sizes $2^{n+k}<500$, the input is the full uniform superposition over all components. All scans use $200$ Monte Carlo trajectories per data point. The performance benchmark of Fig.~\ref{fig:perf} uses the data-loading input of one branch per address with zero bus input, i.e.\ exactly $D=2^n$ branches, and $50$ trajectories; Appendix~\ref{app:perf} documents the measurement and lists the raw values. The section reports first end-to-end benchmarks of the simulator---fidelity, runtime, and exactness---and then trajectory-level tests that verify each theoretical prediction in isolation. As a third, external validation, Appendix~\ref{app:baseline} cross-checks the noise model of the trajectory simulator against an independent circuit-level density-matrix simulation of the complete loading circuit.

\subsection{Benchmarks of the pruned simulator}

The benchmarks start with fidelity, since a speedup only matters once the output is right. Bad branches accumulate as $\mathbb{E}[B]=\mathcal{O}(n^{3}p)\,D$ (Sec.~\ref{sec:algorithm}), so the fidelity should fall with $n$, earlier for stronger noise. Figure~\ref{fig:fidelity} shows the trajectory-averaged fidelity of the qubit encoding for $\varepsilon=\gamma \in \{10^{-5},10^{-4},10^{-3}\}$. At $10^{-5}$ the fidelity stays above $0.988$ for all $n\le10$. At $10^{-3}$ it degrades rapidly beyond $n \approx 4$, consistent with the $\mathcal{O}(n^{3}p)$ bad-branch enlargement of Eq.~\eqref{eq:eb}.

The second question is what the pruning saves. The pruned mode evolves only the bad branches plus one reference branch (Sec.~\ref{sec:algorithm}), so its cost should rise from a single branch at zero noise toward the full $D$ branches as the noise grows. Figure~\ref{fig:perf} compares the full and pruned modes in wall-clock time per query (panel~a), pruned-over-full speedup (panel~b), and explicitly simulated branches (panel~c). At zero noise, the pruned mode evolves exactly the single reference branch, and the speedup reaches $285\times$ at $n=12$. At $\varepsilon=\gamma=10^{-5}$ it explicitly evolves only $3.6\%$--$9.3\%$ of the input branches in panel~c and runs at speedups of $10$--$21\times$ in panel~b. At $10^{-4}$, the bad-branch fraction grows from $12\%$ to $62\%$ with $n$, and the speedup decays from $7.9$ to $1.7$ accordingly. At $10^{-3}$ almost every branch is bad, so the pruned mode approaches the full one, recovering the three-regime structure familiar from the qutrit study of Ref.~\cite{wang2025qram}. The explicitly simulated branch count is also the algorithmic memory of the simulator. The pruned mode stores $\mathcal{O}(B\,2^{k})$ states versus $\mathcal{O}(D\,2^{k})$ for the full mode, and the measured ratios track the bad fraction as expected. Appendix~\ref{app:perf} lists the raw per-setting values and the measurement details.

The last test is exactness. We compare the full and pruned modes before the final tree measurement by reconstructing the predicted groups and matching their complete complex amplitudes, including the residual tree configuration and the input-branch label. We then compare the sampled tree, the normalized output amplitudes, the fidelity, and the random-generator state. All $1000$ trajectory pairs on the performance grid ($n\in\{4,6,8,10,12\}$, $\varepsilon=\gamma\in\{0,10^{-5},10^{-4},10^{-3}\}$, and $50$ seeds per point) agree exactly: every compared quantity coincides to within floating-point roundoff. The previously observed discrepancies originated in reconstruction bookkeeping: when a fired damping jump annihilated the reference good branch, the good groups that were not yet explicitly represented incorrectly retained their nominal weights. The reconstruction now propagates that annihilation to every predicted group. If the reference survives, ordinary reconstruction proceeds. Fixed noise histories test both cases independently of the underlying random-number implementation. This repair introduces no approximation to the predicted evolution. Representative paired fidelities are listed in Appendix~\ref{app:equiv}.

\subsection{Targeted verification of the closed forms and the error structure}

The benchmarks above do not by themselves certify the predictions of Proposition~\ref{prop:qutrit} and Theorems~\ref{thm:multi}--\ref{thm:goodbranch}, so we designed trajectory-level tests for each of them, using damping-only channels throughout. First, for the $H\!\to\!K_0^{2n}\!\to\!H$ formula we run the qubit circuit on a single input branch and select no-jump trajectories, namely those containing no $K_1$ jump. The normalized probability of every output word of the explicitly evolved branch is then compared against Eq.~\eqref{eq:hamming}. Over $n\in\{2,3,4,5,6,8\}$ and $\gamma\in[10^{-5},3\times10^{-2}]$ the agreement holds to floating-point accuracy. Words sharing the Hamming weight $w$ collapse onto a single value, as predicted and shown in Fig.~\ref{fig:verify_hamming}(a). Second, the same test in the qutrit encoding confirms Proposition~\ref{prop:qutrit}. The normalized final state of every good branch is exactly the ideal output. The measured relative weights between branches match the schedule counter $(1-\gamma)^{\Delta c}$ to machine precision, as shown in Fig.~\ref{fig:verify_hamming}(b).

\begin{figure*}[t]
\centering
\includegraphics[width=\textwidth]{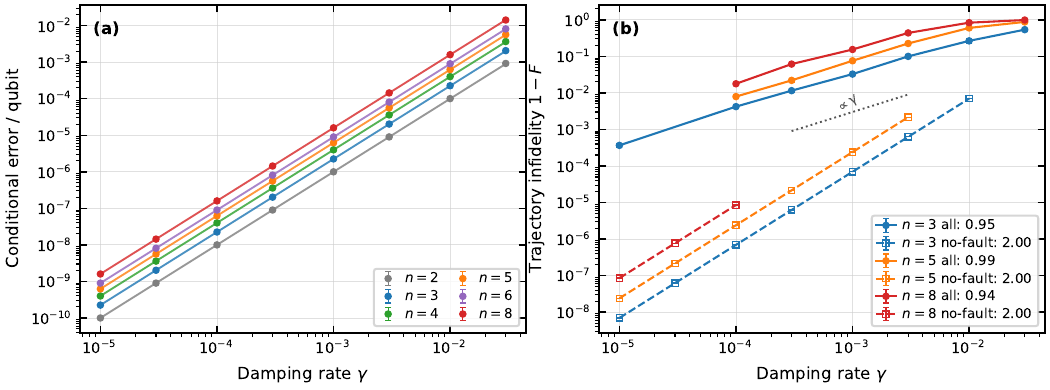}
\caption{Numerical verification of the error structure of Sec.~\ref{sec:errors} with damping-only channels. (a)~Per-data-qubit output error of no-jump trajectories in the qubit encoding, shown as points, versus the prediction $n^2\gamma^2/4$, shown as lines. Log-log fits give slopes $2.000$--$2.003$ and prefactor ratios within $10^{-4}$ of unity at every $n$. (b)~Trajectory-averaged query infidelity, shown as filled points and first order in $\gamma$, versus the same average restricted to fault-free trajectories, shown as open squares and second order. Ensembles containing fewer than $20$ fault events, namely the smallest-$\gamma$ points of the all-trajectory curves, are omitted from the plot and from the fits. Error bars show the standard error of the mean across trajectories.}
\label{fig:verify_scaling}
\end{figure*}

Figure~\ref{fig:verify_scaling}(a) shows the per-data-qubit output error of no-jump trajectories. It tracks $n^2\gamma^2/4$ with fitted slopes $2.000$--$2.003$ and prefactor ratio $1.0001$ at every $n$, confirming the purely second-order coherent leakage of Sec.~\ref{sec:errors}. Figure~\ref{fig:verify_scaling}(b) then separates trajectory ensembles with damping only and $500$ input branches. The trajectory-averaged query infidelity scales linearly in $\gamma$ with fitted slopes $0.94$--$0.99$. The same average restricted to fault-free trajectories scales quadratically, with slopes $2.00$, and lies orders of magnitude lower. All first-order infidelity is therefore carried by trajectories containing faults, as claimed in Sec.~\ref{sec:errors}. Finally, the exhaustive single-fault injection behind Eq.~\eqref{eq:badrange} covers every combination of node, slot, time slice, and fault type on the $n=3$ tree: the tree has $7$ routing nodes, each carrying a routing qubit and a data qubit, so there are $2\times 7=14$ injection positions; the $n=3$, $k=3$ query consists of the $T-1=23$ event-carrying time slices; and both fault types, an $X$ flip and a damping fault, are injected at every position and slice, giving $14\times 23\times 2=644$ cases. Every $X$-type fault reproduces the family-divergence envelope exactly. Of the $322$ damping cases, $178$ are vacuous because no branch carries an excitation in that slot at that slice, so the jump fires with probability zero, and the remaining $144$ stay inside their own subtree. There are zero containment violations.

\section{Conclusion and outlook}
\label{sec:conclusion}

Branch pruning under amplitude damping is sound only because good branches remain predictable under the ever-present no-jump operator. We have shown that this predictability holds in both QRAM encodings, although for sharply different reasons. In the qutrit encoding, the no-jump operator $K_0^{\mathrm{tri}}$ is diagonal and has the wait state as its fixed point, so the no-jump evolution of a good branch reduces to the schedule-computable scalar $a^{c}$, and every good branch ends in the ideal output times one common tree configuration, as stated in Proposition~\ref{prop:qutrit}. In the qubit encoding, the data bus crosses the routing tree in the $X$ basis between two Hadamard walls, so amplitude damping tilts it irreversibly and the scalar predictor does not apply; we proved that the resulting $H\to K_0^{2n}\to H$ structure is nevertheless closed-form predictable (Theorems~\ref{thm:multi} and~\ref{thm:goodbranch}), with per-data-qubit amplitudes governed only by the Hamming weight of the error pattern, Eq.~\eqref{eq:hamming}. The qutrit subtree-containment criterion likewise does not carry over verbatim: its true basis is the configuration uniformity of off-path components, which the wait level protects, and once $X$-type faults leave stuck excitations that migrate up idle ancestors, family divergence moves the bad range of a left-child fault up to its nearest right-child ancestor; the marking rule applies this criterion uniformly to every fault type, staying conservative for damping faults whose damage remains subtree-contained. Together with the address-exposure counter and this corrected criterion, the closed forms and the exact joint sampler of Theorem~\ref{thm:joint} complete a fast noisy simulator for qubit-encoded QRAM that evolves only the bad branches plus one reference branch, loses no accuracy at any order, carries stochasticity only in the explicitly simulated bad branches, and incurs conditional in-branch infidelity of second order only. The numerical experiments validate it at two levels: end-to-end benchmarks driven from identical noise histories agree seed by seed on the end state, with pruned-over-full speedups up to $285\times$, and trajectory-level tests confirm every closed form and the second-order conditional error structure against explicitly evolved branches, to machine precision. 

As an outlook, the wall-to-wall lemma extends naturally to other diagonal no-jump operators, such as the qutrit heating channel, and dephasing, for which the same scalar-counting or commutation arguments apply, is another candidate, so fast pruned simulation should carry over to further noise channels; error-filtration analyses~\cite{lee2023error} can in addition exploit the separation between the first-order infidelity carried by jump-containing trajectories and the second-order leakage inside a no-jump branch. The corrected criterion also carries an architectural lesson: damping faults remain confined to their own subtree, so the wider containment envelope of the qubit encoding comes exclusively from $X$-type strays, and an architecture that suppresses them would restore the tighter qutrit boundary, a property that also matters when the QRAM is embedded into full-stack quantum microarchitectures~\cite{zhou2026hima}. With coherent bucket-brigade routing now demonstrated on superconducting processors~\cite{zhang2025router, shen2025experiment, miao2025router}, the fast simulator completes the classical simulation chain from noisy device to large-scale quantum data access.

\section*{Data availability}
The simulator source code is available at \url{https://github.com/IAI-USTC-Quantum/QRAM-Simulator}.

\begin{acknowledgments}
This work has been supported by the National Key Research and Development Program of China (Grant Nos. 2023YFB4502500 and 2024YFB4504100), the National Natural Science Foundation of China (Grant No. 12404564), and the Anhui Province Science and Technology Innovation (Grant Nos. 202423s06050001 and 202423r06050002).
\end{acknowledgments}

\appendix

\section{Preliminary: the depolarizing pruning simulator of Ref.~\cite{wang2025qram}}
\label{app:prelim}

For self-containedness, this appendix summarizes the pruning simulator of Ref.~\cite{wang2025qram}, which was designed for depolarizing noise and is the starting point of our analysis. Fig.~\ref{fig:prelim} illustrates the data representation, the error sampling on the QRAM tree, and the pruned evolution.

\emph{Data representation.} The simulator never forms a dense state vector. The state of one shot is kept as an unnormalized sparse pure state, a dictionary
\begin{equation}
\ket{\Psi}=\sum_{(i,j,Q)\in\mathcal{S}} c_{i,j,Q}\,\ket{i}_A\ket{j}_D\ket{Q}_{\mathrm{tree}},
\label{eq:sparse}
\end{equation}
of basis components labeled by the address $i$, the data word $j$, and the classical tree configuration $Q$. Every gate of the bucket-brigade schedule is a reversible classical permutation, so each component follows its noiseless orbit independently and the circuit acts on the dictionary entry by entry. The norm is left unnormalized during the evolution and is restored at the end of each damping layer and at the final sampling. Channel-level observables follow from averaging over shots.

\emph{Error sampling.} Each shot samples the complete noise history in advance. For every time slice, each routing element and data qubit independently suffers a depolarizing fault with probability $\varepsilon$. The fault is recorded as a fault event $(v,t,W)$, with node $v$, slice $t$, and sampled Weyl operator $W$. Between fault events the evolution is the deterministic noiseless circuit. At the fault event, $W$ is applied as a unitary to the affected components. Because the routing path of address $i$ visits exactly the ancestors of leaf $i$, a fault at node $v$ can influence only branches whose path crosses $v$, that is, addresses in the leaf range of Eq.~\eqref{eq:subtree}. For amplitude damping, which this work adds on top of the depolarizing scheme, the candidates are drawn over the whole tree and the layer is resolved by the joint sampler of Sec.~\ref{sec:jumps}, with one normalization per layer.

\emph{Pruned simulation.} Given the sampled history, the bad-address union is $\mathcal{R}=\bigcup_{\mathrm{faults}}\mathrm{subtree}(v)$. Branches with $i\in\mathcal{R}$ are bad and are evolved explicitly, fault by fault. Branches with $i\notin\mathcal{R}$ are good. Their tree configuration stays on the noiseless orbit and their amplitude is untouched, so they need not be evolved at all. One reference good branch is still evolved explicitly as a consistency check. With $D$ input branches the expected number of bad branches is $\mathcal{O}(n^2 p)\,D$, which yields the cost $\mathcal{O}(D+p\,\poly(n))$ of the reference simulator. The explicitly evolved bad branches and the untouched good amplitudes are merged into one sparse state, normalized, and sampled.

\emph{Window convention.} The exposure windows of the address counter, Eq.~\eqref{eq:caddr}, are ambiguous by one slice at each edge, i.e.\ whether the first and last exposed slices are counted. We fix the convention once and apply it identically in the schedule counter, the joint sampler, and the explicitly evolved reference branch; all reported quantities refer to this convention.

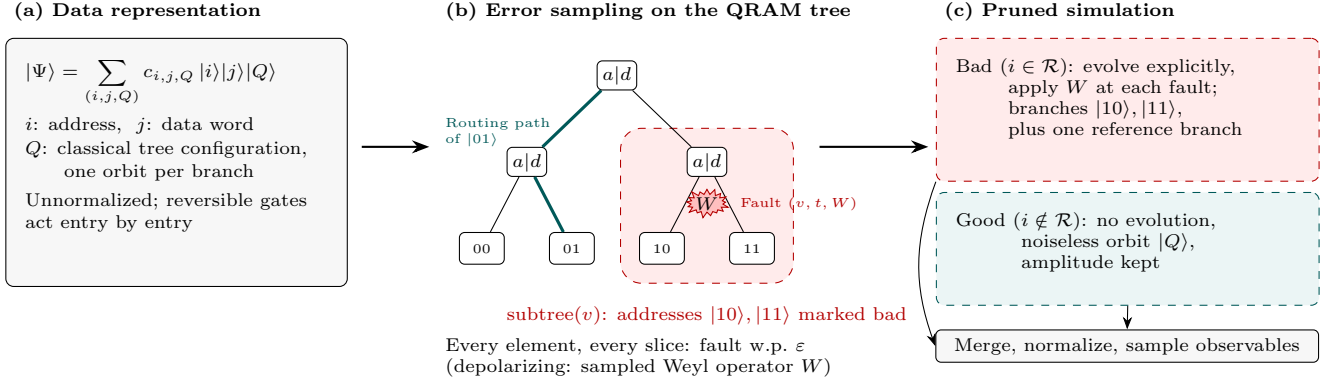
\begin{figure*}[t]
\centering
\resizebox{0.97\textwidth}{!}{%
\begin{tikzpicture}[
  font=\scriptsize,
  box/.style={draw, rounded corners=2pt, align=center, inner sep=3pt},
  qnode/.style={draw, rounded corners=2pt, minimum width=15pt, minimum height=11pt, inner sep=1pt, fill=white},
  leaf/.style={draw, rounded corners=2pt, minimum width=17pt, minimum height=12pt, inner sep=1pt, fill=white},
  faultw/.style={starburst, starburst point height=2.6pt, minimum size=11pt, inner sep=0pt, draw=red!70!black, fill=red!25, font=\scriptsize},
  >=Stealth]
\node[font=\scriptsize\bfseries, anchor=west] at (-0.4, 6.15) {(a) Data representation};
\node[draw, rounded corners=3pt, fill=gray!6, anchor=north west, minimum width=4.55cm, minimum height=3.3cm] (dict) at (-0.4, 5.8) {};
\node[anchor=north west, align=left, font=\scriptsize] at (-0.25, 5.65)
  {$\ket{\Psi}=\displaystyle\sum_{(i,j,Q)} c_{i,j,Q}\,\ket{i}\ket{j}\ket{Q}$\\[3pt]
   $i$: address, \ $j$: data word\\[1pt]
   $Q$: classical tree configuration,\\[1pt]
   \phantom{$Q$: } one orbit per branch\\[3pt]
   Unnormalized; reversible gates\\
   act entry by entry};
\node[font=\scriptsize\bfseries, anchor=west] at (5.35, 6.15) {(b) Error sampling on the QRAM tree};
\node[qnode] (R) at (7.75, 5.3) {$a|d$};
\node[qnode] (A) at (6.55, 4.15) {$a|d$};
\node[qnode] (B) at (8.95, 4.15) {$a|d$};
\node[leaf] (L00) at (5.95, 3.0) {\tiny $00$};
\node[leaf] (L01) at (7.15, 3.0) {\tiny $01$};
\node[leaf] (L10) at (8.35, 3.0) {\tiny $10$};
\node[leaf] (L11) at (9.55, 3.0) {\tiny $11$};
\draw (R) -- (A); \draw (R) -- (B);
\draw (A) -- (L00); \draw (A) -- (L01);
\draw (B) -- (L10); \draw (B) -- (L11);
\draw[very thick, teal!70!black] (R) -- (A) -- (L01);
\node[teal!70!black, font=\tiny, align=left, anchor=west] at (5.35, 4.55) {Routing path\\ of $\ket{01}$};
\node[faultw] (W) at (8.95, 3.6) {$W$};
\node[anchor=west, font=\tiny, red!70!black] at (9.28, 3.6) {Fault $(v,t,W)$};
\begin{scope}[on background layer]
  \node[draw=red!70!black, dashed, rounded corners=8pt, fill=red!8,
        inner xsep=7pt, inner ysep=7pt, fit=(B)(L10)(L11)] (bad) {};
\end{scope}
\node[red!70!black, font=\scriptsize, align=center, anchor=north] at (8.95, 2.35)
  {$\mathrm{subtree}(v)$: addresses $\ket{10},\ket{11}$ marked bad};
\node[align=left, anchor=west, font=\scriptsize] at (5.35, 1.55)
  {Every element, every slice: fault w.p.\ $\varepsilon$\\
   (depolarizing: sampled Weyl operator $W$)};
\node[font=\scriptsize\bfseries, anchor=west] at (12.0, 6.15) {(c) Pruned simulation};
\node[draw=red!70!black, dashed, rounded corners=3pt, fill=red!8, anchor=north west,
      minimum width=5.1cm, minimum height=1.9cm] (badbox) at (12.0, 5.8) {};
\node[anchor=north west, align=left, font=\scriptsize] at (12.15, 5.65)
  {Bad ($i\in\mathcal{R}$): evolve explicitly,\\
   \phantom{Bad (}apply $W$ at each fault;\\
   \phantom{Bad (}branches $\ket{10},\ket{11}$,\\
   \phantom{Bad (}plus one reference branch};
\node[draw=teal!70!black, dashed, rounded corners=3pt, fill=teal!8, anchor=north west,
      minimum width=5.1cm, minimum height=1.5cm] (goodbox) at (12.0, 3.75) {};
\node[anchor=north west, align=left, font=\scriptsize] at (12.15, 3.6)
  {Good ($i\notin\mathcal{R}$): no evolution,\\
   \phantom{Good (}noiseless orbit $\ket{Q}$,\\
   \phantom{Good (}amplitude kept};
\node[box, fill=gray!6, minimum width=5.1cm, font=\scriptsize] (out) at (14.55, 1.7)
  {Merge, normalize, sample observables};
\draw[->] (badbox.south west) to[bend right=20] (out.west);
\draw[->] (goodbox.south) -- (out.north);
\draw[->, thick] (4.35, 4.35) -- (5.25, 4.35);
\draw[->, thick] (10.45, 4.35) -- (11.9, 4.35);
\end{tikzpicture}}
\caption{Schematic of the depolarizing pruning simulator of Ref.~\cite{wang2025qram}.
(a) The state is an unnormalized sparse dictionary of basis components $\ket{i}\ket{j}\ket{Q}$,
one per input branch, acted on entry by entry by the reversible schedule.
(b) Each shot samples the noise history in advance. Every element suffers a fault with
probability $\varepsilon$ per slice, recorded as a fault event $(v,t,W)$ holding the sampled Weyl
operator. The routing path of a branch (teal) visits only the ancestors of its leaf, so a
fault at node $v$ affects only the branches whose paths cross $v$. Its subtree is marked bad
and added to the union $\mathcal{R}$. In the example, the fault on the right child marks
$\ket{10},\ket{11}$, while the branch $\ket{01}$ stays good.
(c) The bad branches and one reference are evolved explicitly, applying the sampled $W$ at
each fault. Good branches keep their noiseless orbit and amplitude. The branches are merged,
normalized, and sampled.}
\label{fig:prelim}
\end{figure*}

\section{End-state equivalence of the full and pruned modes}
\label{app:equiv}

Table~\ref{tab:equiv} lists the end-state fidelities returned by the two modes for three random seeds at four representative operating points. The two numbers agree to machine precision, i.e.\ within a few units in the last place of the double-precision accumulation. This is the expected signature of an exact predictor, in which the pruned mode reconstructs the good branches from the reference branch through the XOR mirror rather than approximating them. Rows with identical fidelities across seeds correspond to trajectories in which no jump was fired, so that only the deterministic no-jump $K_0$ attenuation, identical in both modes, remains.

\begin{table}[h]
\centering
\caption{End-state fidelity of one noisy query, full versus pruned mode, same seed, with identical noise history, jump decisions, and final collapse. $k=3$, $500$ input branches.}
\label{tab:equiv}
\footnotesize
\begin{tabular}{c@{\hspace{4pt}}c@{\hspace{4pt}}c@{\hspace{6pt}}c@{\hspace{6pt}}c@{\hspace{6pt}}c}
\toprule
$n$ & $\varepsilon=\gamma$ & Seed & Full & Pruned & $|\Delta F|$ \\
\midrule
8  & $10^{-5}$ & 880000 & 0.999999914463 & 0.999999914463 & $<10^{-15}$ \\
8  & $10^{-5}$ & 880007 & 0.999999914463 & 0.999999914463 & $<10^{-15}$ \\
8  & $10^{-5}$ & 880014 & 0.999999914463 & 0.999999914463 & $<10^{-15}$ \\
\midrule
10 & $10^{-5}$ & 880000 & 0.999999854463 & 0.999999854463 & $<10^{-15}$ \\
10 & $10^{-5}$ & 880007 & 0.999999854463 & 0.999999854463 & $<10^{-15}$ \\
10 & $10^{-5}$ & 880014 & 0.999999854463 & 0.999999854463 & $<10^{-15}$ \\
\midrule
6  & $10^{-4}$ & 880000 & 0.565825591420 & 0.565825591420 & $<10^{-15}$ \\
6  & $10^{-4}$ & 880007 & 0.999996249505 & 0.999996249505 & $<10^{-15}$ \\
6  & $10^{-4}$ & 880014 & 0.999996249505 & 0.999996249505 & $<10^{-15}$ \\
\midrule
6  & $10^{-3}$ & 880000 & 0.015999567578 & 0.015999567578 & $<10^{-15}$ \\
6  & $10^{-3}$ & 880007 & 0.758346859335 & 0.758346859335 & $<10^{-15}$ \\
6  & $10^{-3}$ & 880014 & 0.999624779533 & 0.999624779533 & $<10^{-15}$ \\
\bottomrule
\end{tabular}
\end{table}

\section{Cross-check against an independent circuit-level simulation}
\label{app:baseline}

The validations of Sec.~\ref{sec:results} and Appendix~\ref{app:equiv} compare two modes of the same simulator driven from identical noise histories. This appendix adds an external reference: the complete qubit-encoded loading circuit unfolded gate by gate on an ordinary qubit register and evolved as a density matrix under explicitly applied channels---no trajectory sampling, no sparse-tree bookkeeping, and no shared state representation or noise-channel implementation with the branch simulator; the time-slice schedule alone is read from the same scheduler. The register carries $n+k+2(2^{n}-1)$ qubits, so the cross-check is restricted to the smallest noisy instance $n=2$, $k=1$ ($9$ qubits, memory $0110$, data-loading input of one branch per address with zero bus), over the full operating range of both noise channels.

\emph{Encoded circuit.} The register consists of the $n$ address qubits, the $k$ bus qubits, and, for each routing node $v$, a routing qubit $a_v$ and a data qubit $d_v$. Every logical operation of the bucket-brigade schedule becomes an ordinary gate sequence: the first-layer address copy is a CNOT from address bit $n\!-\!1\!-\!\ell$ into $d_0$; the bus enters and leaves through SWAPs with $d_0$, wrapped in Hadamard walls on the bus; the internal swap of layer $\ell$ is $\mathrm{SWAP}(a_v,d_v)$ on every node of the layer; the controlled swap of layer $\ell$ acts on every node as $\mathrm{SWAP}(d_v,d_{2v+1})$ controlled by $a_v=0$ and $\mathrm{SWAP}(d_v,d_{2v+2})$ controlled by $a_v=1$; and the memory fetch is a phase kickback, $\pi$ rotations on the leaf data qubits controlled by the leaf routing qubit for every set memory bit. The time-slice schedule is read from the scheduler that also drives the trajectory simulator, so the cross-check pins down the encoded form of each operation, the channel placement, and the averaging---while the noise-free output distributions of the two simulators agree exactly (to machine precision).

\emph{Channel placement.} The channels act only on the $2(2^{n}-1)$ tree qubits; the address register and the bus stay noise-free, as in the model of Sec.~\ref{sec:noise}. In a time slice whose routing front spans $\ell$ layers, each of the $N_\ell=2(2^{\ell}-1)$ active tree positions carries the depolarizing channel with marginal probability $p$ per slice, matching the binomial fault draw of the trajectory sampler. Depolarizing is the uniform $\{X,Z,Y\}$ mixture of total strength $\varepsilon$, the bit-phase-flip element being the unitary $ZX=iY$. The damping layer applies the Kraus operators of Eq.~\eqref{eq:kraus} site by site as fixed matrices, with no state-dependent weights and no post-selection; both sides of the comparison are trace preserving and evaluated at unit trace.

\emph{Conventions.} The ideal end state carries, per address $a$, the bus word $b\oplus m[a]$ and the tree in its ground state. On the trajectory side, $2000$ shots yield the per-trajectory fidelity of the normalized state after the final tree measurement, averaged over shots, together with its Monte Carlo standard error; on the density side the same quantity is exact, $F_{\rho}=\langle\psi_{\mathrm{id}}|\rho|\psi_{\mathrm{id}}\rangle$. Distribution agreement is measured by the classical fidelity $F_{\mathrm{cls}}=(\sum_q\sqrt{pq})^{2}$ and the total-variation distance of the renormalized joint (address, bus) output distributions, and the states themselves by the Frobenius distance $\|\Delta\rho\|_{F}$. Every quantity is evaluated at unit trace on both sides, so each pair of like quantities should agree within the Monte Carlo error.

Table~\ref{tab:base-circuit} lists the result. At every operating point the trajectory ensemble and the density matrix agree within the Monte Carlo scale: the fidelity difference is at most $7\times10^{-3}$ and the output distributions agree at $\mathrm{TVD}\le0.014$, with both sides carrying unit trace. For a sample mean of normalized pure projectors the estimated root-mean-square Frobenius error is $\sqrt{(1-\operatorname{Tr}\overline{\rho}^{\,2})/(N_{\mathrm{shots}}-1)}$, and the observed $\|\Delta\rho\|_{F}$ values are compatible with this scale.

\begin{table}[h]
\centering
\caption{Circuit-level cross-check: trajectory simulator ($2000$ shots) versus density matrix applying the fixed Kraus matrices of Eq.~\eqref{eq:kraus}. $n=2$, $k=1$, memory $0110$, data-loading input. Both sides are trace preserving.}
\label{tab:base-circuit}
\footnotesize
\begin{tabular}{lcccc}
\toprule
Noise & TVD & $F_{\rho}$ & $\overline{F}_{\mathrm{traj}}$ & $\|\Delta\rho\|_{F}$ \\
\midrule
Pauli $p=0.02$      & 0.0017 & 0.7263 & 0.7285 & 0.0089 \\
Damping $\gamma=0.05$ & 0.0128 & 0.5759 & 0.5759 & 0.0163 \\
Mixed $p=\gamma=0.02$ & 0.0123 & 0.5841 & 0.5768 & 0.0164 \\
Damping $\gamma=0.2$  & 0.0133 & 0.1676 & 0.1636 & 0.0157 \\
\bottomrule
\end{tabular}
\end{table}

As a second, fully external check, the same circuit description and sampled histories are re-run with the \texttt{uniqc}/QuTiP implementation~\cite{johansson2012qutip}. Six fixed noise histories reproduce the trajectory simulator's output probabilities to machine precision, entry by entry. Averaging $2000$ trajectories at $\gamma=0.05$ gives $\mathrm{TVD}=0.0128$ and classical fidelity $0.9996$ against the density side, and the mixed setting $p=\gamma=0.02$ gives $\mathrm{TVD}=0.0123$ and $0.9996$ likewise; both implementations report unit trace. The exported histories carry the sampler tag \texttt{joint\_auxiliary\_whole\_tree\_v1} and the layer-normalization convention of Sec.~\ref{sec:jumps}; the experiment is packaged in the \texttt{qram-simulator} Python package and re-runs with a single command at fixed seed.

\section{Experiment details and raw data of the performance benchmark}
\label{app:perf}

This appendix documents how the performance benchmark of Fig.~\ref{fig:perf} is produced and lists the raw values behind its panels. Each data point pairs a noise level $\varepsilon=\gamma\in\{0,10^{-5},10^{-4},10^{-3}\}$ with an address size $n\in\{4,6,8,10,12\}$, at $k=3$ data bits and the data-loading input of one branch per address, so the full mode evolves exactly $D=2^{n}$ branches. Every setting runs $50$ independent noise histories, and the two modes of a setting share the per-trajectory seed, so the pair is driven by the identical fault history, jump decisions, and final collapse; only the execution differs.

The wall-clock time per query is read from a monotonic clock around two stages: the execution of the query, which contains the sampling of the fault history, the branch evolution, and, in the pruned mode, the analytical substitution of the good branches; and the sampling of the output register that produces the fidelity estimate. The speedup of panel~(b) is the ratio of the mean wall-clock time of the full mode to that of the pruned mode over the same $50$ trajectories. The evolved fraction of panel~(c) is the number of explicitly evolved branches, the bad branches plus the single reference branch, divided by $D=2^{n}$. The benchmark is compiled with MSVC in the Release configuration and runs single-threaded on an Intel Core i9-13900KF under Windows~11.

Table~\ref{tab:perfraw} lists the per-setting means; the three panels of Fig.~\ref{fig:perf} plot exactly these columns.

\begin{table}[h]
\centering
\caption{Raw performance data of Fig.~\ref{fig:perf}: mean wall-clock time per query in milliseconds over $50$ trajectories, the pruned-over-full speedup (the ratio of the two time columns), and the fraction of input branches evolved explicitly.}
\label{tab:perfraw}
\footnotesize
\begin{tabular}{c@{\hspace{4pt}}c@{\hspace{6pt}}r@{\hspace{6pt}}r@{\hspace{6pt}}r@{\hspace{6pt}}r}
\toprule
$\varepsilon=\gamma$ & $n$ & $t_{\mathrm{full}}$ & $t_{\mathrm{pruned}}$ & Speedup & Evolved (\%) \\
\midrule
0 & 4 & 0.5731 & 0.0307 & 18.7 & 6.25 \\
0 & 6 & 4.043 & 0.0859 & 47.1 & 1.56 \\
0 & 8 & 21.56 & 0.2424 & 88.9 & 0.39 \\
0 & 10 & 124.8 & 0.7592 & 164 & 0.10 \\
0 & 12 & 725.1 & 2.548 & 285 & 0.02 \\
\midrule
$10^{-5}$ & 4 & 0.7496 & 0.0497 & 15.1 & 6.25 \\
$10^{-5}$ & 6 & 4.815 & 0.2348 & 20.5 & 3.88 \\
$10^{-5}$ & 8 & 29.12 & 1.417 & 20.6 & 3.61 \\
$10^{-5}$ & 10 & 173.9 & 11.29 & 15.4 & 5.28 \\
$10^{-5}$ & 12 & 1258 & 125.4 & 10.0 & 9.32 \\
\midrule
$10^{-4}$ & 4 & 0.7586 & 0.0964 & 7.87 & 11.75 \\
$10^{-4}$ & 6 & 4.885 & 0.8590 & 5.69 & 16.28 \\
$10^{-4}$ & 8 & 31.26 & 8.250 & 3.79 & 24.04 \\
$10^{-4}$ & 10 & 243.0 & 95.92 & 2.53 & 38.34 \\
$10^{-4}$ & 12 & 2750 & 1655 & 1.66 & 61.54 \\
\midrule
$10^{-3}$ & 4 & 0.7895 & 0.3668 & 2.15 & 45.00 \\
$10^{-3}$ & 6 & 6.251 & 4.491 & 1.39 & 69.28 \\
$10^{-3}$ & 8 & 55.15 & 51.81 & 1.06 & 93.08 \\
$10^{-3}$ & 10 & 800.1 & 794.1 & 1.01 & 99.03 \\
$10^{-3}$ & 12 & 11814 & 11816 & 1.00 & 99.89 \\
\bottomrule
\end{tabular}
\end{table}

\end{document}